\documentclass[aps,pra,twocolumn,showpacs,superscriptaddress]{revtex4-1}

\usepackage[utf8]{inputenc}
\usepackage[english]{babel}
\usepackage[T1]{fontenc}
\usepackage{amsmath, amsfonts, amssymb, amsthm, mathrsfs, amssymb,  bbm, bm,graphicx,times,txfonts,color,subfigure}
\usepackage{hyperref}

\usepackage{tikz}
\usetikzlibrary{shapes.geometric, arrows}

\newtheorem{definition}{Definition}
\newtheorem{theorem}{Theorem}
\newtheorem{lemma}{Lemma}
\newtheorem{corollary}{Corollary}
\newtheorem{property}{Property}

{\noindent{\textbf{Remark.}}}{}

{\noindent{\textbf{Exercise.}}}{}

{\noindent{\textbf{Example.}}}{}

\newcommand{\CR}{{{\mathscr{C}}_{\mathcal{R}}}}

\newcommand{\CL}{{{\mathscr{C}}_{\ell_1}}}
\newcommand{\AR}{{{\mathscr{A}}_{\mathcal{R}}}}

\newcommand*{\cD}{\mathcal{D}}
\newcommand*{\cE}{\mathcal{E}}

\newcommand*{\cH}{\mathscr{H}}

\newcommand*{\cL}{\mathcal{L}}

\newcommand*{\cU}{\mathcal{U}}

\newcommand{\bC}{\mathbb{C}}

\newcommand*{\ket}[1]{|#1\rangle}
\newcommand*{\bra}[1]{\langle #1|}
\newcommand*{\braket}[2]{\langle #1| #2 \rangle}
\newcommand*{\proj}[1]{\ket{#1}\!\bra{#1}}
\newcommand*{\Tr}{\mathrm{Tr}}
\newcommand*{\tr}{\mathrm{Tr}}

\newcommand{\beq}{\begin{equation}}
\newcommand{\eeq}{\end{equation}}

\newcommand{\I}{\openone}

\begin{document}

\title{Robustness of asymmetry and coherence of quantum states}

\author{Marco~Piani}
\email{marco.piani@strath.ac.uk}
\affiliation{SUPA and Department of Physics, University of Strathclyde, Glasgow G4 0NG, UK}

\author{Marco~Cianciaruso}
\email{cianciaruso.marco@gmail.com}
\affiliation{Dipartimento di Fisica ``E. R. Caianiello'', Universit\`a degli Studi di Salerno, Via Giovanni Paolo II, I-84084 Fisciano (SA), Italy; and INFN Sezione di Napoli, Gruppo Collegato di Salerno, Italy}
\affiliation{School of Mathematical Sciences, The University of Nottingham, University Park, Nottingham NG7 2RD, UK}

\author{Thomas~R.~Bromley}
\email{thomas.r.bromley@gmail.com}
\affiliation{School of Mathematical Sciences, The University of Nottingham, University Park, Nottingham NG7 2RD, UK}

\author{Carmine~Napoli}
\email{c.napoli15@studenti.unisa.it}
\affiliation{Dipartimento di Fisica ``E. R. Caianiello'', Universit\`a degli Studi di Salerno, Via Giovanni Paolo II, I-84084 Fisciano (SA), Italy; and INFN Sezione di Napoli, Gruppo Collegato di Salerno, Italy}
\affiliation{School of Mathematical Sciences, The University of Nottingham, University Park, Nottingham NG7 2RD, UK}

\author{Nathaniel~Johnston}
\email{nathaniel@njohnston.ca}
\affiliation{Department of Mathematics and Computer Science, Mount Allison University, Sackville, New Brunswick E4L 1E2, Canada}

\author{Gerardo~Adesso}
\email{gerardo.adesso@nottingham.ac.uk}
\affiliation{School of Mathematical Sciences, The University of Nottingham, University Park, Nottingham NG7 2RD, UK}

\begin{abstract}
Quantum states may exhibit asymmetry with respect to the action of a given group. Such an asymmetry of states can be considered as a resource in applications such as quantum metrology, and it is a concept that encompasses quantum coherence as a special case. We introduce explicitly and study the robustness of asymmetry, a quantifier of asymmetry of states that we prove to have many attractive properties, including efficient numerical computability via semidefinite programming, and an operational interpretation in a channel discrimination context. We also introduce the notion of asymmetry witnesses, whose measurement in a laboratory detects the presence of asymmetry. We prove that properly constrained asymmetry witnesses provide lower bounds to the robustness of asymmetry, which is  shown to be a directly measurable quantity itself. We then focus our attention on coherence witnesses and the robustness of coherence, for which we prove a number of additional results; these include an analysis of its specific relevance in phase discrimination and quantum metrology, an analytical calculation of its value for a relevant class of quantum states, and tight bounds that relate it to another previously defined coherence monotone.
\end{abstract}

\pacs{03.65.Ud, 03.67.Bg, 03.67.Ac, 03.65.Ta}

\date{March 1, 2016}
\maketitle


\section{Introduction}

Symmetry is a central concept in physics, as it imposes constraints and allows simplifications in the study of properties and evolutions of physical systems. It has a vast range of applicability, from particle physics to cosmology, up to its elevation to the status of a principle on which a physical theory can be based~\cite{gross1996role}. Symmetry is defined with respect to the action of a symmetry group. A quantum state, described by a density operator, may or may not be invariant under the action of the group. The extent to which the symmetry is broken by the quantum state constitutes its degree of \emph{asymmetry}. The advent of quantum information processing has fostered the study of asymmetry~\cite{Vaccaro2003,Bartlett2007,Vaccaro2008,Gour2008,Gour2009,Marvian2013,Marvian2014a,Marvian2015}, on one hand because of its potential applications as a resource in quantum communication and estimation tasks, and on the other hand because of the availability of conceptual and technical tools developed in quantum information theory, which can be efficiently borrowed to characterize more rigorously the notions of symmetry and asymmetry.

Quantum metrology is one of the areas of quantum information processing more readily deployable in real-world scenarios, and is drawing a large international effort to exploit effectively quantum features like superposition and entanglement for enhanced sensing technologies~\cite{Giovannetti2011}.  Related to metrology is also the study of quantum reference frames~\cite{Vaccaro2003,Bartlett2007,Vaccaro2008}. The understanding of quantum reference frames and of their manipulation is instrumental in harnessing the advantages promised by quantum communication, and fundamental in building a fully consistent quantum picture of nature, overcoming the need for the notion of a classical system of reference. The presence and degree of asymmetry in a state of a quantum system allows one to distinguish between the action of different  elements of the group, making such an asymmetry key in quantum communication and quantum metrology~\cite{Giovannetti2011,Hall2012}, and rendering the system at hand a potential reference frame~\cite{Vaccaro2003,Bartlett2007,Vaccaro2008}. Finally, in a context where the physical evolution is constrained by a symmetry, the asymmetry of (the state of) a system can allow one to overcome the limitations imposed by the symmetry group and to perform transformations and measurements on other systems that would otherwise be forbidden, like preparing those systems in states that violate the symmetry constraints. In a wide spectrum of situations, the asymmetry of a state can be seen therefore as a resource, that allows one to perform tasks, be them passive---as detection in metrology---or active---as in the manipulation of quantum systems.

In this article, which is also the companion to~\cite{PRL}, we introduce explicitly and study the {\it robustness of asymmetry}, a measure of asymmetry of quantum states that we prove to have a number of attractive properties. As a tool on the way, but relevant on its own right, we also introduce the notion of asymmetry witness, that is, the general notion of an observable whose measured value provides qualitative and---under suitable constraints---quantitative information about the asymmetry of a state. We then specialize our analysis to the {\it robustness of coherence}, complementing the dedicated study in \cite{PRL}. Although quantum coherence, understood as the superposition of orthogonal ``classical'' states, can be seen as a particular case of asymmetry (namely with respect to the group of time translations generated by a Hamiltonian diagonal in the basis of such ``classical'' states), it deserves in fact particular focus for the following reasons. First,  it can be considered as the most essential quantum feature exhibited by a single system; second, it underpins all forms of quantum correlations in composite systems; third, it can be related to quantum enhancements in diverse instances of quantum information, thermodynamics, condensed matter physics, and life sciences \cite{Aaberg2006,Lloyd2011,Huelga2013,Baumgratz2014,Grosshans2003,Giovannetti2011}. Our work thus directly contributes to advancing the recently initiated program for a rigorous operational  characterization of quantum coherence as a resource~\cite{Aaberg2006,Marvian2013,Baumgratz2014,Levi2014,Marvian2015,Streltsov2015b,Du2015,Hall2015,India2015,Winter2015,MaxCoh,Streltsov2015,PRL}.

The present paper is organized as follows. In Section~\ref{sec:restheoryasymmetry} we recall the notions of symmetric and asymmetric quantum states, respectively non-resources and resources in a resource theoretic approach to asymmetry. In Section~\ref{sec:asymmetry} we study in general the robustness of asymmetry by: defining it (Section~\ref{sec:asymmetrydefinition}); proving its fundamental properties (Section~\ref{sec:asymmetryproperties}); proving that its evaluation can be cast as a semidefinite program (Section~\ref{sec:asymmetrySDP}); introducing the notion of asymmetry witnesses  and providing (observable) bounds (Section~\ref{sec:asymmetrywitnesses}); and providing an operational interpretation in terms of advantage for channel discrimination tasks  (Section~\ref{sec:asymmetrychannel}). In Section~\ref{sec:robusteness_of_coherence} we focus on the robustness of coherence, presenting the details of the results announced in \cite{PRL}, in particular obtaining explicit and analytical bounds on it (Section~\ref{sec:bounds_coherence}), and calculating it analytically for relevant cases (Section~\ref{sec:exact_coherence}).

\section{Resource theory of asymmetry}
\label{sec:restheoryasymmetry}
The notion of {\it asymmetry} with respect to a given representation of a symmetry group $\sf G$ has been recently studied in quantum mechanics adopting the information theoretic paradigm of {\it resource theories} \cite{Gour2008,Gour2009,Marvian2013,Marvian2014a,Marvian2015}.
In general, the overall objective of any resource theory is to understand and formalize the quantification and manipulation of a given physical phenomenon, in order to facilitate its exploitation in the most efficient way \cite{Coecke2014}. This framework can be applied even beyond the domain of physical sciences \cite{Renner2015}.

In quantum mechanics, any resource theory is defined by the (typically convex) set of {\it free states}, and by a set of {\it free operations}, see e.g.~\cite{Coecke2014,Brandao2015}. The free states are states not possessing the resource under consideration, while any non-free state can be defined as a {\it resource state}, or shortly a resource. On the other hand, the free operations are defined so that they are unable to create the resource from free states, that is, they must map the set of free states into (a subset of) itself. Depending on the context and framework, some additional limitations may or may not be taken into account when defining the free operations.
A typical example of a resource theory is the theory of entanglement in composite quantum systems \cite{Plenio1997,Horodecki2009}, where free states are identified as separable (i.e., unentangled) states, and free operations are conventionally taken to be local operations and classical communication (LOCC), which nonetheless form a proper subset of the maximal set of all possible operations mapping separable states into separable states, and even of the set of the so-called separable operations~\cite{bennettnonlocality}.

Once free states and free operations are defined, the main aim of a resource theory resides in the study of the manipulation of the resource by the (chosen) free operations~\cite{Coecke2014,Brandao2015}. We remark that in this paper our concern lies mainly in the quantification of a resource---asymmetry---and not so much in its manipulation. Nonetheless, we do refer to the notion of free operations for a meaningful reason.

One of the merits of a resource theory framework is indeed that it naturally leads to a set of conditions which should be satisfied by any proposed quantifier of the resource. In particular, any valid resource measure should vanish on the set of free states (and is termed faithful if it vanishes only on such set), and should be nonincreasing under the chosen free operations: given that the latter are unable to create resources from free states, they should also be unable to increase the resource content of non-free states. Any resource measure which obeys such a fundamental constraint can be regarded as a {\it resource monotone} \cite{Vidal2000}. Additionally, it is often demanded that a resource measure be convex, i.e.~nonincreasing under mixing, if the set of free states is convex. Once a resource theory is established, therefore, it proves useful to validate any proposal for a resource measure.

This will be precisely the case for the robustness of asymmetry, on which this paper is focused. To proceed, we first recall the main ingredients that define the resource theory of asymmetry \cite{Marvian2013}.

\subsection{Symmetric states as free states}
\label{sec:restheoryasymmetrystates}

Given a Hilbert space $\cH$ and the convex set ${\mathscr{D}}(\cH)$ of density operators acting on it, let us consider a symmetry group ${\sf G}$ with associated unitary representation  $\{U_g\}_{g \in {\sf G}}$ on $\cH$.
Let us define the action of $U_g$ on a state $\xi \in {\mathscr{D}}(\cH)$  in terms of the superoperator $\cU_g$ as follows,
\begin{equation}\label{eq:cU}
\cU_g(\xi) = U_g \xi U_g^\dagger\,.
\end{equation}
A state $\sigma \in {\mathscr{D}}(\cH)$ is defined as {\it symmetric} with respect to $\sf G$ if and only if
\begin{equation}\label{eq:defsym}
\cU_g(\sigma) =\sigma\,,
\end{equation}
for all $g\in {\sf G}$.
Notice that this is equivalent to the condition
$\cE(\sigma)=\sigma$,
with
\begin{equation}\label{eq:cE}
\cE(\xi)=\frac{1}{|{\sf G}|}\sum_{g \in {\sf G}} \cU_g (\xi)
\end{equation}
denoting the average of the action of the group \cite{Gour2009}.

We indicate by
\begin{equation}\label{eq:S}
\mathscr{S} := \{\sigma  \in {\mathscr{D}(\cH)} : \cE(\sigma)=\sigma\}
 \end{equation}
 the set of all symmetric states $\sigma$ according to the above definition. This constitutes the set of free states for the resource theory of asymmetry \cite{Marvian2013}, and it is evidently convex.
Any other state $\rho \in {\mathscr{D}}(\cH)$ is {\it asymmetric} with respect to $\sf G$, that is, is a resource state. Explicitly, $\rho$ is asymmetric if and only if there exists a $g\in {\sf G}$ such that
\begin{equation}\label{eq:defasym}
\cU_g(\rho) \neq \rho\,.
\end{equation}
Equivalently, $\rho$ is asymmetric if and only if $\cE(\rho)\neq\rho$.

\subsection{An example of free operations: covariant operations}
\label{sec:restheoryasymmetryoperations}
As mentioned, we will not be particularly concerned with the manipulation of asymmetry. For this reason, we do not need to be very specific about the class of free operations. Furthermore, the quantity we set out to study, the robustness of asymmetry (and later, more specifically, the robustness of coherence) turns out to be a resource monotone in a very general sense (see Section~\ref{sec:asymmetryproperties} for more details). Nonetheless, for concreteness, we provide an example of free operations which have been adopted for the resource theory of asymmetry. This is the set of {\it covariant operations} with respect to the group $\sf G$ (or, in short, $\sf G$-covariant operations) \cite{Marvian2013}. Any such operation is defined by a superoperator $\cL: \mathscr{D}(\mathscr{H}) \rightarrow \mathscr{D}(\mathscr{H})$ such that,
\begin{equation}\label{eq:Gcov}
\cL \big(\cU_g(\xi)\big)=\cU_g\big(\cL(\xi)\big)\,,\quad \forall \ g \in {\sf G}, \ \xi \in  \mathscr{D}(\mathscr{H})\,.
\end{equation}
Equivalently, any covariant operation is defined by $[\cL, \cU_g]=0$, $\forall \ g \in {\sf G}$.

\section{Robustness of asymmetry}
\label{sec:asymmetry}
In this section we define and investigate a quantifier of the asymmetry of quantum states with respect to a group representation $\{U_g\}_{g \in {\sf G}}$, in compliance with the resource theory formalism introduced in the previous section.

\subsection{Definition}
\label{sec:asymmetrydefinition}

\begin{definition}[\bf Robustness of asymmetry]
Given a state $\rho \in \cD(\cH)$, we define the {\it robustness of asymmetry} (RoA) of $\rho$ as
\begin{equation}\label{eq:RoA}
\AR(\rho) = \min_{\tau \in \mathscr{D}(\cH)} \left\{ s\geq 0\ \Big\vert\ \frac{\rho + s\ \tau}{1+s} =: \sigma \in {\mathscr{S}}\right\}\,,
\end{equation}
that is, as the minimum weight $s$, parametrized as in \eqref{eq:RoA}, of another state $\tau$, such that its normalized convex mixture with $\rho$ results into a symmetric state $\sigma$.
\end{definition}

If we denote by $s^\star$ the value of $s$ achieving the minimum in Eq.~(\ref{eq:RoA}), with corresponding states $\tau^\star$ (a generic state, not necessarily symmetric) and $\sigma^\star$ (a symmetric state), then $\AR(\rho) = s^\star$, and
\begin{equation}\label{eq:pseudo}
\rho = \big(1+\AR(\rho)\big)\sigma^\star - \AR(\rho)\tau^\star
\end{equation}
is said to realize an optimal pseudomixture for $\rho$.

It is immediate to realize that $\AR(\rho)$ can also be characterized as
\beq
\label{eq:RoAequivalent}
\AR(\rho) = \min_{\sigma \in {\mathscr{S}}} \left\{ s\geq 0\ \Big\vert\ \rho \leq (1+ s)\,\sigma \right\}\,.
\eeq
This follows since Eq.~\eqref{eq:pseudo} implies $\rho\leq  \big(1+\AR(\rho)\big)\sigma^\star$, with $\sigma^\star\in {\mathscr{S}}$, which means that $\AR(\rho)$ is lower-bounded by the minimum on the right-hand side of Eq.~\eqref{eq:RoAequivalent}. On the other hand, suppose $\rho \leq (1+ s)\sigma$ for some $\sigma\in{\mathscr{S}}$. Then we can write
\[
\sigma = \frac{\rho+s\tau}{1+s}
\]
with $\tau=\big[(1+s)\sigma-\rho\big]/s$ a valid state. This proves that the minimum in Eq.~\eqref{eq:RoAequivalent} is also an upper bound for $\AR(\rho)$, henceforth we conclude that \eqref{eq:RoAequivalent} holds.

Notice that the robustness of a resource can be defined for any general resource theory~\cite{Brandao2015}. Previously, robustness quantifiers have been studied for entanglement, steering-type correlations, non-locality and even correlations beyond quantum~\cite{Robustness,GenRobustness,Piani2015,GellerPiani}. Nonetheless, up to our knowledge, the notion of robustness of asymmetry has not been explored yet.

\subsection{Properties}
\label{sec:asymmetryproperties}

Here we prove that the RoA satisfies a number of properties which qualify it as a valid {\it asymmetry monotone}. We remark that the properties listed here are valid, with suitable adaptations, for all measures of robustness defined in a resource theoretic context \cite{Brandao2015}, with respect to a convex set of free states (in our case, symmetric states) that is closed under a chosen set of free operations (in our case, e.g., covariant operations). The first such example of a robustness measure was defined for entanglement theory \cite{Robustness,GenRobustness}. Most of the proofs reported here are in fact straightforward translations of those originaly produced for the robustness of entanglement.

\begin{property}
\label{A1}
The RoA is bounded as
\beq
0\leq\AR(\rho)\leq \dim(\cH)-1\,
\eeq
for any $\rho\in\cD(\cH)$. Furthermore the RoA is faithful, that is
\begin{equation}\label{eq:A1}
\AR(\rho) = 0\ \ \Longleftrightarrow\ \ \rho \in \mathscr{S}\,.
\end{equation}
\end{property}
\begin{proof}
That  $\AR(\rho)\geq 0$ and that $\AR(\rho)= 0$ if and only if $\rho\in{\mathscr{S}}$ is evident by definition (\ref{eq:RoA}). Let $d=\dim(\cH)$. The bound $\AR(\rho)\leq d -1$ is proven by considering that the maximally mixed state $\I/d$ is symmetric for any unitary representation on $\cH$, and that
\[
\rho\leq\I=(1+(d-1))\frac{\I}{d}\,,
\]
for every $\rho\in\cD(\cH)$. We get the claim by comparing this with Eq.~\eqref{eq:RoAequivalent}.
\end{proof}

\begin{property}
\label{A2}
Let $\{\Gamma_l\}_{l=1}^{m}$ be an instrument, that is, a collection of $m$ completely positive subchannels, summing up to a completely positive trace preserving channel $\cL(\rho) = \sum_{l=1}^{m} \Gamma_l(\rho)$, such that $\Gamma_l(\sigma)/\tr[\Gamma_l(\sigma)] = \sigma_l \in \mathscr{S}$, $\forall \ l=1,\ldots,m$, and for any $\sigma \in \mathscr{S}$.
Then, the RoA is monotonically nonincreasing on average under $\{\Gamma_l\}_{l=1}^{m}$:
\begin{equation}\label{eq:A2}
\AR(\rho)\geq \sum_l \Tr[\Gamma_l(\rho)] \AR\Big(\frac{\Gamma_l(\rho)}{\Tr[\Gamma_l(\rho)]}\Big)\,.
\end{equation}
\end{property}
\begin{proof}
Let $\tau^\star$ and $\sigma^\star$ denote the (generic and symmetric, respectively) states in the optimal pseudomixture for $\AR(\rho)$ as in Eq.~(\ref{eq:pseudo}), and let us apply the subchannel $\Gamma_l$ on both sides, so that
\[
\Gamma_l(\rho) = \big(1+\AR(\rho)\big)\Gamma_l(\delta^\star) - \AR(\rho)\Gamma_l(\tau^\star)\,.
\]
By defining
\begin{equation*}
\begin{split}
\sigma_l&=\frac1{(1+s_l)}\frac1{p_l}\left(1+  \AR\left(\rho\right)\right) \Gamma_l(\delta^\star),\\
\tau_l&=\frac1{s_l} \frac1{p_l} \AR\left(\rho\right) \Gamma_l(\tau^\star),\\
s_l&= \frac1{p_l} \AR\left(\rho\right)\tr\left[ \Gamma_l(\tau^\star)\right],
\end{split}
\end{equation*}
with $p_l=\tr[\Gamma_l(\rho)]$, we can write
\[\rho_l=\left(1+s_l\right)\sigma_l-s_l\tau_l,
\]
where $\rho_l=\Gamma_l(\rho)/p_l$.
Since the latter pseudomixture for each $\rho_l$ is not necessarily optimal, it follows by Eq.~\eqref{eq:RoAequivalent}  that
\[\AR\left(\rho_l\right)\le s_l.\]
Taking the weighted average over all subchannels, and recalling that $\sum_l \tr[\Gamma_l(\xi)]=1$ for any state $\xi$, we finally get
\[
\sum_l p_l \AR\left(\frac{\Gamma_l(\rho)}{p_l}\right) \leq \sum_l \frac{p_l}{p_l} \AR\left(\rho\right)\tr\left[ \Gamma_l(\tau^\star)\right] = \AR(\rho).
\]
%
\end{proof}
Notice that this property is expressed in very general terms: If one has only one subchannel equal to a channel ($m=1$), then Eq.~(\ref{eq:A2}) proves standard monotonicity under free operations that do not create the resource, e.g., under covariant operations. If on the other hand one identifies each subchannel with a Kraus operator, i.e. $\Gamma_l(\rho) = K_l \rho K_l^\dagger$ with $\sum_{l=1}^m K_l^\dagger K_l = \openone$, then Eq.~(\ref{eq:A2}) proves the stronger monotonicity under selective operations \cite{Plenio1997,Baumgratz2014}.

\begin{property}
\label{A3}
The RoA is convex, that is
\begin{equation}
\label{eq:A3}
\AR\big(p \rho_1 + (1-p) \rho_2\big) \leq  p \AR(\rho_1) + (1-p) \AR(\rho_2)\,,
\end{equation}
for any probability $p\in [0,1]$, and any states $\rho_1, \rho_2 \in \mathscr{D}(\mathscr{H})$.
\end{property}
\begin{proof}
Let $\rho_1$ and $\rho_2$ be two states, and consider for each the optimal pseudomixture as in Eq.~(\ref{eq:pseudo}),
\[\rho_k = \big(1+\AR(\rho_k)\big)\delta_k^\star - \AR(\rho_k) \tau_k^\star\,,\]
 with $k=1,2$.
 Take now the convex combination
 \[
 \rho = p \rho_1 + (1-p) \rho_2\,,\]
 with $p \in [0,1]$, and notice that a nonoptimal pseudomixture of the form $\rho = (1+s) \sigma - s \tau$ can be written, with
\begin{equation*}
\begin{split}
\sigma &=\frac{1}{1+s}\left[p\big(1+\AR(\rho_1)\big)\delta^\star_1 + (1-p)\big(1+\AR(\rho_2)\big)\delta^\star_2\right],\\
\tau&= \frac{1}{s}\left[p \AR(\rho_1) \tau^\star_1 + (1-p) \AR(\rho_2) \tau^\star_2\right],\\
s&=p \AR(\rho_1) + (1-p) \AR(\rho_2),
\end{split}
\end{equation*}
By definition, the optimal pseudomixture for $\rho$ in the definition of the RoA will have $\AR(\rho)=s^{\star} \leq s$, which proves Eq.~(\ref{eq:A3}).
\end{proof}



\subsection{Robustness of asymmetry as a semidefinite program}
\label{sec:asymmetrySDP}

We now show that the evaluation of the RoA can be recast as a semidefinite program (SDP)~\cite{vandenberghe1996semidefinite}. In the Supplemental Material \cite{epapsA} we provide a MATLAB~\cite{MATLAB:2015} code to evaluate such an SDP for any input state $\rho$ and any group representation $\{U_g\}$, using the open-source MATLAB-based modeling system for convex optimization CVX~\cite{cvx,gb08}.

\begin{theorem}
\label{thm:main}
The RoA $\AR(\rho)$ corresponds to the SDP
\beq
\label{eq:SDProbustenssprimal}
\begin{aligned}
\min&{} 		&  \quad&\Tr[\tilde{\sigma}]-1\\
\textup{s.t.}&{} 	&	 &\tilde{\sigma}\geq \rho, \\
{}&{}			& &\cE(\tilde{\sigma})=\tilde{\sigma}.
\end{aligned}
\eeq
Strong duality holds, and the RoA can be equivalently calculated via the dual SDP
\begin{equation}\label{eq:SDP}
\begin{aligned}
\max&{} 		&  \quad&-\Tr[W\rho]\\
\textup{s.t.}&{} 	&	 &W\leq\I,\\
{}&{}			& &\cE(W)\geq0,
\end{aligned}
\end{equation}
where $W$ is a Hermitian operator on $\cH$, and the  SDP constraint in the last line of \eqref{eq:SDP} can be restricted to $\cE(W)=0$, that is, the achieved maximum is the same in both cases.
\end{theorem}
\begin{proof}
By incorporating the factor $(1+s)$ appearing in \eqref{eq:RoAequivalent} into the unnormalized state $\tilde{\sigma}=(1+s)\sigma$,
we can reexpress $\AR(\rho)$ as the SDP \eqref{eq:SDProbustenssprimal}.
It is immediate to check that strong duality holds, since a feasible solution of the primal SDP is $\tilde{\sigma}=(1+\epsilon)\openone$, for $\epsilon>0$.

The SDP can be cast in the standard form~\cite{watrous2009}
\beq
\label{eq:SDProbustenssprimalstandard}
\begin{aligned}
\min&{} 		&  \quad&\Tr[C\tilde{\sigma}]-1\\
\textup{s.t.}&{} 	&	 &\Lambda(\tilde{\sigma})\geq B, \\
{}&{}			& & \tilde{\sigma}\geq 0,
\end{aligned}
\eeq
with
\[
C=\I,\ \Lambda(\tilde\sigma)
=
\begin{pmatrix}
\tilde{\sigma} 			& 			0 					& 0 \\
0 					& \cE(\tilde{\sigma})- \tilde{\sigma}		& 0 \\
0					&			0					&  -\cE(\tilde{\sigma}) + \tilde{\sigma}
\end{pmatrix}, \mbox{ and }
B
=
\begin{pmatrix}
\rho		 			& 			0 					& 0 \\
0 					& 			0					& 0 \\
0					&			0					&  0
\end{pmatrix}.
\]
The dual SDP is then~\cite{watrous2009}
\[
\begin{aligned}
\max&{} 		&  \quad&\Tr[B Y]-1\\
\textup{s.t.}&{} 	&	 &\Lambda^\dagger(Y)\leq C, \\
{}&{}			& & Y \geq 0,
\end{aligned}
\]
where (the asterisks indicate irrelevant submatrices)
\[
Y
=
\begin{pmatrix}
Y_1		 			& 			* 					& * \\
* 					& 			Y_2					& * \\
*					&			*					&  Y_3
\end{pmatrix},
\]
and
\[
\begin{aligned}
\Lambda^\dagger(Y)&= Y_1+(\cE(Y_2)-Y_2) - (\cE(Y_3)-Y_3)\\
&=Y_1+\cE(Y_2-Y_3)-(Y_2-Y_3).
\end{aligned}
\]
The dual SDP then simplifies to
\[
\begin{aligned}
\max&{} 		&  \quad&\Tr[\rho Y_1]-1\\
\textup{s.t.}&{} 	&	 &Y_1+\cE(Y_2-Y_3)-(Y_2-Y_3)\leq \I, \\
{}&{}			& & Y_1,Y_2,Y_3\geq0.
\end{aligned}
\]
Noticing that the target function can only be larger if the first condition is saturated with equality, that is $Y_1 = \I -\cE(Y_2-Y_3)+(Y_2-Y_3)$, and that the latter expression depends  on $Y_2$ and $Y_3$ only through the combination $\tilde{W}=Y_2-Y_3$, which is unrestricted with respect to positivity, we arrive at
\beq
\label{eq:newSDPasymmetry0}
\begin{aligned}
\max&{} 		&  \quad&-\Tr[(\cE(\tilde{W})-\tilde{W})\rho]\\
\textup{s.t.}&{} 	&	 &\cE(\tilde{W})-\tilde{W} \leq \I .\\
\end{aligned}
\eeq
Using the fact that $\cE$ is idempotent (equivalently, a superoperator acting as a projector, that is, $\cE^2(\xi) = \cE(\xi)$ for all states $\xi$), it is easy to see that this is equivalent to
\beq
\label{eq:newSDPasymmetry1}
\begin{aligned}
\max&{} 		&  \quad&-\Tr[W\rho]\\
\textup{s.t.}&{} 	&	 &W \leq \I, \\
		&{}	&	&\cE(W)=0.
\end{aligned}
\eeq
Indeed, let $\tilde{W}$ be such that $\cE(\tilde{W})-\tilde{W} \leq \I$. Then, if one defines $W:=\cE(\tilde{W})-\tilde{W}$, one has $\cE(W)=0$. Thus the value \eqref{eq:newSDPasymmetry0} is a lower bound for the value \eqref{eq:newSDPasymmetry1}. On the other hand, let $W$ be such that $\cE(W)=0$, and define $\tilde{W}=-W$. Then $\cE(\tilde{W})-\tilde{W}=W$. Thus the value \eqref{eq:newSDPasymmetry1} is a lower bound for the value \eqref{eq:newSDPasymmetry0}.

In turn, one sees easily that \eqref{eq:newSDPasymmetry1} is equivalent to (\ref{eq:SDP}), i.e., the optimum in \eqref{eq:SDP} is achieved by a  $W$ that satisfies $\cE(W)=0$. That \eqref{eq:SDP} is an upper bound for \eqref{eq:newSDPasymmetry1}  is obvious. On the other hand, take $W$ such that $\cE(W)\geq 0$, and consider $W'=W-\cE(W)$, which  satisfies by definition $\cE(W')=0$. Since $\cE(W)\geq 0$, one has $-\Tr[W'\rho]\geq -\Tr[W\rho]$  and $W'\leq W \leq \I$. Thus,  \eqref{eq:newSDPasymmetry1} is an upper bound for \eqref{eq:SDP}.\end{proof}

The reason that in Theorem \ref{thm:main} we refer to the SDP \eqref{eq:SDP} rather than \eqref{eq:newSDPasymmetry1} is mostly the fact that the condition $\cE(W)\geq0$ is more robust than the condition $\cE(W)=0$, both numerically and experimentally. That is to say that, for example, directly measuring an observable $W$ satisfying $\cE(W)\geq0$, for the purpose of asymmetry detection and estimation, is experimentally feasible, while meeting the exact condition $\cE(W)=0$ is impossible in practice (although it might be considered feasible if we are content with implementing the condition within error bars). We discuss further the issue of practically measuring or estimating the RoA in Section~\ref{sec:asymmetrywitnesses}.

For later convenience, we also report an alternative form of the dual of the SDP in Eq.~(\ref{eq:newSDPasymmetry1}), rewritten as
\beq
\label{eq:SDPrestated}
\begin{aligned}
\max&{} 		&  \quad&\Tr[X\rho]-1\\
\textup{s.t.}&{} 	&	 &X\geq 0,  \\
{}&{}			& &\cE(X)=\I,
\end{aligned}
\eeq
where we have simply made the substitution $X=\I-W$ for the SDP variable.

\subsection{Asymmetry witnesses and observable lower bounds to the robustness of asymmetry}
\label{sec:asymmetrywitnesses}

Here we follow up from the previous (rather technical) subsection by presenting some insightful physical remarks stemming from Theorem~\ref{thm:main}, and in particular Eq.~\eqref{eq:SDP}.
We first observe that, thanks to the fact that $\Tr[\cE(Y)X]=\Tr[Y\cE(X)]$ for all $X,Y$, the condition $\cE(W)\geq 0$ is equivalent to
\begin{equation}\label{eq:witness}
\Tr[W\sigma]\geq 0\,,\quad \forall \ \sigma \in \mathscr{S}\,.
\end{equation}
This means that any Hermitian operator $W$ such that $\cE(W)\geq 0$ can be regarded as an {\it asymmetry witness}, in analogy with the theory of entanglement witnesses \cite{Horodecki2009}. For any such $W$, finding  $\Tr[W\rho]<0$ implies that the state $\rho$ is asymmetric, that is, a resource.

The SDP formulation in Theorem~\ref{thm:main} further implies that
\begin{equation}\label{eq:lowitness}
\max\{0,\,-\tr[\rho W]\} \leq \AR(\rho)\,,
\end{equation}
for all the asymmetry witnesses $W$ subject to the constraints of Eq.~(\ref{eq:SDP}).
By the same statement, it follows that for any state $\rho$ there exists an optimal (state-dependent) witness $W^\star$ such that the RoA of $\rho$ is exactly {\it observable}   as
\begin{equation}\label{eq:observable}
\AR(\rho) = -\tr[\rho W^\star]\,.
\end{equation}
These observations entail that the RoA can be regarded as an instance of a quantitative asymmetry witness, in analogy to quantitative entanglement witnesses \cite{auden2006,Brandao2005,QWitness1,QWitness2,QWitness3,GuhneToth}.

By employing suboptimal witnesses $W$ in Eq.~(\ref{eq:lowitness}), e.g., tailored on experimental capabilities, one can estimate the RoA from below. We can now readily provide a chain of explicit lower bounds to the RoA of an arbitrary state $\rho$, as follows.

\begin{theorem}
\label{prp:lowerbounda}
For any $\rho \in \mathscr{D}(\mathscr{H})$, it holds that
\beq
\label{eq:asymmetrybound}
\AR(\rho) \geq \frac{\|\rho-\cE(\rho)\|_2^2}{\|\cE(\rho)\|_\infty}\geq\frac{\|\rho-\cE(\rho)\|_2^2}{\|\cE(\rho)\|_2}\geq \|\rho-\cE(\rho)\|_2^2\,,
\eeq
where $\|\xi\|_p$ denotes the Schatten $p$-norm of an operator $\xi$,
\begin{equation}\label{eq:pnorm}
\|\xi\|_p = \Big(\tr \big[|\xi|^p\big]\Big)^{\frac1p}\,,
\end{equation}
with $\|\xi\|_\infty$ amounting to the largest singular value of $\xi$ (also known as operator norm), and $\|\xi\|_2 = \sqrt{\tr [\xi^\dagger \xi]}$ reproducing the Hilbert-Schmidt norm of $\xi$.
\end{theorem}

\begin{proof}
Notice first that the witness
\begin{equation}\label{eq:witness2}
W=\frac{\cE(\rho) - \rho}{\|\cE(\rho)\|_\infty}
\end{equation}
is by construction an admissible operator in Eq.~\eqref{eq:newSDPasymmetry1}, as
\[\frac{\cE(\rho) - \rho}{\|\cE(\rho)\|_\infty}\leq \frac{\cE(\rho) }{\|\cE(\rho)\|_\infty}\leq \I\]
and
\[\cE(W)=\frac{\cE^2(\rho) - \cE(\rho)}{\|\cE(\rho)\|_\infty}=0.\]
Thus, by Eq.~(\ref{eq:lowitness}), we have
\[
\begin{aligned}
\AR(\rho)
&\geq -\Tr[W\rho]\\
&=-\frac{\Tr[(\cE(\rho)-\rho)\rho]}{\|\cE(\rho)\|_\infty}\\
&=-\frac{\Tr[\cE(\rho)^2]-\Tr[\rho^2]}{\|\cE(\rho)\|_\infty}\\
&=\frac{\Tr[\rho^2]-\Tr[\cE(\rho)^2]}{\|\cE(\rho)\|_\infty}\\
&=\frac{\|\rho-\cE(\rho)\|_2^2}{\|\cE(\rho)\|_\infty},\\
\end{aligned}
\]
having used that $\Tr[\rho\cE(\rho)]=\Tr[\cE(\rho)^2]$.

The second inequality in \eqref{eq:asymmetrybound} is due to the hierarchical relation $\|\xi\|_\infty\leq\|\xi\|_2$ for any operator $\xi$, and the third inequality is due to the fact that $\|\cE(\rho)\|_2\leq\|\cE(\rho)\|_1=\Tr[\rho]=1$.
\end{proof}

We remark that \[\|\rho-\cE(\rho)\|_2^2= \Tr[\rho^2]-\Tr[\cE(\rho)^2]\] and that
\[\|\cE(\rho)\|_2=\sqrt{\Tr[\cE(\rho)^2]}.\]
Both  $\Tr[\rho^2]$ and $\Tr[\cE(\rho)^2]$ are directly measurable, for an unknown $\rho$, as long as one can prepare and perform measurements on  two identical and independent copies $\rho^{\otimes 2}$ of $\rho$.
It is known \cite{ekert2002direct} that one has in fact
\[\Tr[\rho^2]=\Tr[\rho^{\otimes 2} V]\,,\]
and
\[\Tr[\cE(\rho)^2]=\Tr[(\cE(\rho)\otimes\cE(\rho)) V]=\Tr[\rho^{\otimes 2}\cE^{\otimes 2}(V)]\,,\]
where $V$ is the swap operator acting on two copies of the Hilbert space $\mathscr{H}$ of the system, \[
V\ \ket{\psi}\otimes  \ket{\phi} = \ket{\phi} \otimes \ket{\psi}\,,\]
 for any $\ket{\psi}, \ket{\phi}  \in \mathscr{H}$.


Immediate lower bounds to the RoA of a state can also be obtained based on the measurement of \emph{any} set of observables $\{O_i\}$, $i=1,\ldots,k$, conveying the expectation values $o_i=\tr[O_i\rho]$, and not necessarily tailored to the measurement of the RoA. Indeed,
one can consider asymmetry witnesses of the form $W=\sum_{i=1}^k c_i O_i + m\openone$, for $c_1,\ldots,c_k,m \in\mathbb{R}$, and lower bound the RoA by the SDP
(code available~\cite{epapsA})
\begin{subequations}
\label{eq:witnessSDP}
\begin{align}
\textrm{max}\quad&-\left(\sum_{i=1}^k c_i o_i + m\right)\\
\textrm{s.t.}\quad &\sum_{i=1}^k c_i O_i + m\openone \leq \openone,\\
		& \cE\left(\sum_{i=1}^k c_i O_i + m\openone\right)\geq 0.
\end{align}
\end{subequations}
The SDP \eqref{eq:witnessSDP} provides the best possible asymmetry witness based on the available data: a witness that, in a sense, we can at least measure ``on paper'' (or rather, on computer) given the actual measurements performed in the lab.

One can make even better use of available experimental data, by exactly estimating the minimal RoA compatible with the data. This goes  beyond trying to construct the best asymmetry witness out of the data, but, remarkably, can also be cast as an SDP, more precisely as
\begin{subequations}
\label{eq:physicalSDP}
\begin{align}
\textrm{min} \quad&\Tr[\tilde{\sigma}]-1\\
\textrm{s.t.}\quad&\tilde{\sigma}\geq \rho,  \\
				& \cE(\tilde{\sigma})=\tilde{\sigma}, \\
 &  \rho\geq 0, \quad \Tr[\rho]=1,\quad \tr[O_i\rho] = o_i\quad \forall i. \label{eq:physicalconsistentcy}	
\end{align}
\end{subequations}
The SDP \eqref{eq:physicalSDP} is essentially the same as the primal SDP \eqref{eq:SDProbustenssprimal} for the RoA, but it does not assume the knowledge of the underlying state $\rho$ of which we want to know the asymmetry. Instead, the SDP constraints \eqref{eq:physicalconsistentcy} impose the minimal condition that a physical state $\rho$ exists that is compatible with the observed data, in the spirit of \cite{auden2006}.
In general, one cannot think of the estimate of RoA given by  \eqref{eq:physicalSDP} as resulting from calculating the expectation value of just one asymmetry witness that is accessible with the available data. One can easily argue that there is a single witness that would give the same numerical result, but in general we cannot assume that we have measured it, or the ability of directly reconstructing its expectation value from the available data. The  SDP \eqref{eq:physicalSDP} instead exploits the full knowledge about the individual measurements.

Albeit the estimate of the RoA given by \eqref{eq:physicalSDP} is in principle always better than the estimate given by \eqref{eq:witnessSDP}, there are reasons to consider \eqref{eq:witnessSDP} of interest, and potentially prefer it to \eqref{eq:physicalSDP}. One is that, as mentioned, the output of the SDP \eqref{eq:witnessSDP} comprises the best single asymmetry witness that we can have knowledge of based on the data; the knowledge of such a witness can then be used in devising ways to exploit the asymmetry of the state, as in the case where we use it for metrology (see Section~\ref{sec:asymmetrychannel} for an example of this). Another reason deals instead with what could be considered a kind of ``fragility'' of \eqref{eq:physicalSDP}: indeed, the latter SDP also acts as a feasibility test for the compatibility of the measurement results with a physical scenario---the existence of a physical (normalized and positive semidefinite) state that gives rise to the statistics. The issue is that the data collected could be incompatible with a physical state, in the sense of not satisfying \eqref{eq:physicalconsistentcy}, because of experimental errors / imprecisions / statistics, and hence do not lead to any reasonable lower bound to the RoA. Another way of looking at it, is that \eqref{eq:physicalSDP}, while perfectly well defined from an abstract point of view, can only be used in practice when there is some assurance that the data are (or have been processed to be) compatible with some physical state. How to best do this while obtaining a certifiable lower bound to the RoA is beyond the scope of the present work.


\subsection{Robustness of asymmetry as advantage in covariant channel discrimination games}
\label{sec:asymmetrychannel}

In this section we provide a general operational interpretation for the RoA in the context of discriminating  quantum channels.
Given as usual a unitary representation $\{U_g\}$ of a group ${\sf G}$, let us  consider the unitary channels $\mathcal{U}_g$ whose action is defined in Eq.~(\ref{eq:cU}). Suppose we want to discriminate among the set of such channels, which can be applied to an input probe state $\rho$ with an {\it a priori} probability distribution $\{p_g\}_{g \in {\sf G}}$ on the group. We can think of this process as a game, in which a message is encoded on the probe system initialized in $\rho$ by the action of one such channel $\cU_g$, and the aim of the game is to guess correctly which $\cU_g$ was implemented, hence decoding the message. To do so, one needs to measure the output state $\cU_g(\rho)$ after the channel, by means of a positive operator-valued measure (POVM) $\{M_g\}$, with elements satisfying $M_g\geq 0$, $\sum_g M_g =\I$.

For any given measurement strategy, we can define the probability of success $p^{\textup{succ}}$, that is, the probability of guessing correctly in the discrimination (or, equivalently, of decoding the message), as
\begin{equation}
{p^{\textup{succ}}_{\{p_g\},{\{M_g\}}}}(\rho)=\sum_g p_g \Tr[\cU_g(\rho)M_g]\,.
\end{equation}
As indicated, the probability of success depends on the prior probability distribution $\{p_g\}$, on the choice of POVM, and, most importantly, on the probe state $\rho$ on which the information is encoded. We can further define
\begin{equation}\label{eq:psucc}
{p^{\textup{succ}}_{\{p_g\}}}(\rho) = \max_{\{M_g\}} {p^{\textup{succ}}_{\{p_g\},{\{M_g\}}}}(\rho)\,,
\end{equation}
as the optimal probability of success for a given $\rho$ and prior $\{p_g\}$, maximized over all possible POVMs used in the discrimination/decoding.

We can now distinguish the cases where the input probe state is symmetric (a free state) or asymmetric (a resource state).
In the first case, let the probe state be denoted by $\sigma \in \mathscr{S}$. Since by definition $\cU_g(\sigma) = \sigma$, no information is actually encoded in the state during the process. We have then,
\begin{equation}\label{eq:psuccsympre}
\begin{aligned}
p^{\textup{succ}}_{\{p_g\},\{M_g\}}(\sigma)
&=\sum_g p_g \Tr[\cU_g(\sigma)M_g]\\
&=\sum_g p_g \Tr[\sigma M_g]\\
&\leq p_{\sf G}^{\max} \Tr\left[\sigma \sum_g  M_g\right]\\
&=p_{\sf G}^{\max}\,,
\end{aligned}
\end{equation}
where we have defined the maximal a priori probability $p_{\sf G}^{\max}=\max_{g\in {\sf G}} p_g$. The upper bound in Eq.~(\ref{eq:psuccsympre}) can be achieved by a strategy consisting in always guessing the group element $g^{\max}$ with the highest associated prior probability $p_{g^{\max}} \equiv p_{\sf G}^{\max}$, that is, by implementing a POVM with $M_{g^{\max}}=\I$, and $M_g = 0$ $\forall \  g \neq g^{\max}$. Since this is independent of the specific symmetric state $\sigma$, we have then that the optimal probability of success for any symmetric state $\sigma \in \mathscr{S}$, as defined in Eq.~(\ref{eq:psucc}), is given by
\begin{equation}\label{eq:psuccsym}
p^{\textup{succ}}_{\{p_g\}}(\mathscr{S}):=p^{\textup{succ}}_{\{p_g\}}(\sigma) = p_{\sf G}^{\max} = \max_{g\in {\sf G}} p_g\,.
\end{equation}

It is clear that, by using an asymmetric state $\rho$ as a probe, one can expect to achieve in general a higher probability of success than $p^{\textup{succ}}_{\{p_g\}}(\mathscr{S})$ in the above channel discrimination game: in other words, asymmetry is expected to be a useful resource for the considered task. One can then wonder precisely {\it how larger} an optimal success probability can be reached exploiting asymmetric probes, compared to  symmetric probes. We now address this question by showing that it is precisely the RoA of $\rho$ which determines the advantage enabled by choosing $\rho$ as a probe in the above channel discrimination game, as opposed to any symmetric probe $\sigma$. This provides an intuitive and general operational interpretation for the RoA.

\begin{theorem}\label{thm:operational}
For any state $\rho$ and any prior probability distribution $\{p_g\}_{g \in {\sf G}}$ it holds that
\begin{multline}
\label{eq:boundsright}
\max\left\{\frac{1}{|{\sf G}|}(1+\AR(\rho))\,,\,p^{\textup{succ}}_{\{p_g\}}(\mathscr{S})\right\}\\
\leq
p^{\textup{succ}}_{\{p_g\}}(\rho)
\leq \\
(1+\AR(\rho))p^{\textup{succ}}_{\{p_g\}}(\mathscr{S}).
\end{multline}
\end{theorem}

\begin{proof}
The second inequality is just a consequence of the definition of $\AR(\rho)$, which implies that there is a symmetric $\sigma$ such that
$\rho\leq (1+\AR(\rho))\sigma$,
so that
\[
\begin{aligned}
\sum_g p_g \Tr[\cU_g(\rho)M_g]
& \leq (1+\AR(\rho)) \sum_g p_g \Tr[\sigma M_g] \\
& \leq (1+\AR(\rho)) p^{\textup{succ}}_{\{p_g\}}(\mathscr{S}).
\end{aligned}
\]

On the other hand, to prove the first inequality, consider the optimal $X$ for the SDP \eqref{eq:SDPrestated}, which is such that $\Tr[X\rho]=1+\AR(\rho)$. We observe that, since $X\geq 0$, also $\cU_g(X) \geq 0$. Furthermore, due to Eq.~(\ref{eq:cE}), and since $\cE(X)=\I$, we have that $\{M_g\}$, with
\begin{equation}\label{eq:validpovm}
M_g = \frac{1}{|{\sf G}|} \cU_g(X)\,,
\end{equation}
is a valid POVM, as $\sum_g M_g = \cE(X)=\openone$. We have then
\[
\begin{aligned}
\sum_g p_g \Tr[M_g  \mathcal{U}_g(\rho)]
&=\sum_g \frac{p_g}{|{\sf G}|}\Tr[U_g X U_g^{\dagger} U_g \rho U_g^{\dagger}]\\
&=\sum_g \frac{p_g}{|{\sf G}|} \Tr[X \rho]\\
&=\frac{1}{|{\sf G}|}\Tr[X \rho]\\
&=\frac{1}{|{\sf G}|}(1+\AR(\rho)).
\end{aligned}.
\]
This proves that
\[
p^{\textup{succ}}_{\{p_g\}}(\rho) \geq \frac{1}{|{\sf G}|}(1+\AR(\rho))\,.
\]
The other possibility in the lower bound of (\ref{eq:boundsright}), i.e.~$p^{\textup{succ}}_{\{p_g\}}(\rho) \geq p^{\textup{succ}}_{\{p_g\}}(\mathscr{S})$,  follows from the fact that simply guessing $g^{\textup{max}}$ is always a potentially valid strategy.
\end{proof}

We remark that the proof of Theorem~\ref{thm:operational} also provides a proof of the general relation
\[
\AR(\rho) \leq |{\sf G}| - 1\,,
\]
since the probability of success is bounded above by 1.

As a consequence of Theorem~\ref{thm:operational}, one can furthermore write the following explicit result.

\begin{corollary}\label{cor:operational}
For any state $\rho$ and prior probability distribution $\{p_g\}_{g \in {\sf G}}$, it holds that
\beq
\max_{\{p_g\}} \frac{p^{\textup{succ}}_{\{p_g\}}(\rho)}{p^{\textup{succ}}_{\{p_g\}}(\mathscr{S})}=1+\AR(\rho)\,.
\eeq
\end{corollary}
\begin{proof}
Let us divide \eqref{eq:boundsright} by $p^{\textup{succ}}_{\{p_g\}}(\mathscr{S})$. We then find
\begin{multline}
\max\left\{\frac{1}{|{\sf G}|\ p^{\textup{succ}}_{\{p_g\}}(\mathscr{S})}(1+\AR(\rho))\,,1\right\}\\
\leq
\frac{p^{\textup{succ}}_{\{p_g\}}(\rho)}{p^{\textup{succ}}_{\{p_g\}}(\mathscr{S})}
\leq \\
(1+\AR(\rho)).
\end{multline}
The lower bound matches the upper bound in the case $p^{\textup{succ}}_{\{p_g\}}(\mathscr{S})=p_G^{\textup{max}} =\frac{1}{|{\sf G}|}$, that is for a flat prior probability distribution over ${\sf G}$.
\end{proof}

We notice that, although we have focused on the discrimination of the channels $\{\cU_g\}$, the result of Corollary  \ref{cor:operational} can be generalized to any set of channels of which $\{\cU_g\}$ is a subset. That is, we have the following:
\begin{corollary}\label{cor:operationalgeneral}
Let $\{\Lambda_i\}$ be a set of channels such that $\{\cU_g\}\subseteq \{\Lambda_i\}$. Consider a probability distribution $\{p_i\}$ on  $\{\Lambda_i\}$. Then
\beq
\label{eq:operationalgeneral}
\max_{\{p_i\}} \frac{p^{\textup{succ}}_{\{p_i\}}(\rho)}{p^{\textup{succ}}_{\{p_i\}}(\mathscr{S})}=1+\AR(\rho)\,.
\eeq
Here $p^{\textup{succ}}_{\{p_i\}}(\mathscr{S})$ is the maximal probability of success in the discrimination of the channels $\{\Lambda_i\}$ when these are applied with a priori probability distribution $\{p_i\}$, maximized over the choice of any input arbitrary symmetric state; $p^{\textup{succ}}_{\{p_i\}}(\rho)$ is instead the maximal probability of success by using the given $\rho$.
\end{corollary}
\begin{proof}
As in the proof of Theorem \ref{thm:operational}, it is easy to see that
\[
p^{\textup{succ}}_{\{p_i\}}(\rho)= \max_{\{M_i\}}\sum_i p_i \Tr[\Lambda_i(\rho)M_i] \leq (1+\AR(\rho))p^{\textup{succ}}_{\{p_i\}}(\mathscr{S})
\]
is a direct consequence of $\rho\leq (1+\AR(\rho))\sigma$. Thus, in general
\[
\frac{p^{\textup{succ}}_{\{p_i\}}(\rho)}{p^{\textup{succ}}_{\{p_i\}}(\mathscr{S})}\leq 1+\AR(\rho).
\]
From the fact that $\{\cU_g\}\subseteq \{\Lambda_i\}$ and from Corollary \ref{cor:operational} we know that this upper bound can be saturated.
\end{proof}

We remark that, in general, for a set of channels $\{\Lambda_i\}$ such that $\{\cU_g\}\subseteq \{\Lambda_i\}$, the best probability of success by using a symmetric state, $p^{\textup{succ}}_{\{p_i\}}(\mathscr{S})$, is not just equal to $\max_i p_i$, as the channels $\{\Lambda_i\}$ act in general non-trivially even on symmetric states. Yet, Eq.~\eqref{eq:operationalgeneral} holds. The point is that the case $\{\Lambda_i\}=\{\cU_g\}$, with $p_g=1/|{\sf G}|$, is the one where there is the largest possible advantage in using the asymmetric state; any class of channel discrimination problems that include this latter problem will satisfy Eq.~\eqref{eq:operationalgeneral}.

\subsection{Finite versus continuous groups: the $\textsf{U}(1)$ case}

While we referred to finite groups so far in the paper, the theory we have developed can be applied also to compact groups like, e.g., the $d$-dimensional representation of $\textsf{U}(1)$,  or the defining representation of $\textsf{U}(d)$, as well as its tensor product representation 
on $n$ $d$-dimensional systems.

Consider the representation $\{U(g)\}$ of a continuous compact group ${\sf G}$ with Haar measure $dU(g)$, e.g., a compact Lie group. Suppose one  can define a finite set ${\sf X}=\{U_x\}\subseteq \{U(g)\}$ such that, for any state $\xi$,
\[
\cE(\xi)=\int_{\sf G} U(g)\xi U^\dagger(g)dU(g)=\frac{1}{|{\sf X}|}\sum_x U_x \xi U_x^\dagger.
\]
This applies, e.g., to the defining representation of $\textsf{U}(d)$, in which case ${\sf X}$ is said to form a $1$-design~\cite{gross2007}.
Then, when we discuss the discrimination of covariant channels as in Sec.~\ref{sec:asymmetrychannel}, we can simply imagine discriminating among the action of the discrete set of channels $U_x$ in ${\sf X}$.

On the other hand, for the sake of the definition of the RoA, and its computation by means of an SDP, it is clear that the only thing that matters is the actual action of the group average $\cE$, independently of any detail of how $\cE$ is implemented.

In the case of $\textsf{U}(1)$, its $d$-dimensional representation is given by $\{U(\theta)\}$ with \begin{equation}
\label{eq:Utheta}
 U(\theta)=\sum_{j=0}^{d-1} {\rm e}^{i{\theta}j}\proj{j},\quad\theta\in[0,2\pi].
 \end{equation}
The group average is equivalent to the total dephasing in the basis $\{\ket{j}\}$:
\begin{equation}
\cE(\xi)=\frac{1}{2\pi} \int_0^{2\pi} d\theta U(\theta) \xi U(\theta)^\dagger = \sum_{j=0}^{d-1} \proj{j} \xi \proj{j}=:\Delta(\xi).
\label{eq:dephasing}
\end{equation}
Equivalently,  the group average is the same as the average over the representation of the cyclic group ${\sf G'}={\sf Z}_d$, with $|{\sf G'}|=d$, on $\bC^d$, given by $\{Z^k\}_{k=0}^{d-1}$, where $Z$ is the phase flip operator
\begin{equation}\label{eq:phase}
Z\ket{j}={\rm e}^{i\frac{2\pi}{d}j}\ket{j}.
\end{equation}
In this case, asymmetry with respect to such a $d$-dimensional representation of ${\sf U}(1)$ can be regarded as {\it coherence}, that is, quantum superposition with respect to the reference basis $\{\ket{j}\}_{j=0}^{d-1}$ in the Hilbert space $\mathscr{H} = \bC^d$ \cite{Aaberg2006,Marvian2013,Baumgratz2014,Marvian2015,PRL}.

The rest of the paper is devoted to analyzing the specifics of this case.


\section{Robustness of coherence}
\label{sec:robusteness_of_coherence}

A resource theory for quantum coherence can be constructed as a special case of the resource theory of asymmetry, when we consider (a)symmetry with respect to $\textsf{U}(1)$, or, alternatively, as discussed at the end of the previous section, with respect to (the $d$-dimensional representation of) the cyclic group ${\sf Z}_d$.

Specializing from Eq.~(\ref{eq:S}), and using the notation of Eq.~(\ref{eq:dephasing}), we indicate by
\begin{equation}\label{eq:I}
\mathscr{I} := \{\delta  \in {\mathscr{D}(\cH)} : \Delta(\delta)=\delta\}
 \end{equation}
 the set of {\it incoherent states}, which can be regarded as the free states for the resource theory of coherence.
 Equivalently, every incoherent  state is diagonal in the reference basis,
 \begin{equation}\label{eq:incoherent}
 \delta=\sum_j \delta_j \proj{j}.
 \end{equation}

 The specialization of Eq.~(\ref{eq:Gcov}) using the representation in Eq.~(\ref{eq:Utheta}) defines instead the set of translationally invariant operations, which can be considered as one possible choice of free operations for the resource theory of coherence, within the context of asymmetry \cite{Marvian2015}. However, there have been different proposals to define free operations for coherence not derived from the asymmetry framework, see e.g.~\cite{Aaberg2006,Baumgratz2014,YadinG2015,Streltsov2015}. We discuss such approaches in more detail in the companion Letter \cite{PRL}, to which we refer for additional insights and motivations regarding the study of quantum coherence as a resource. Here we only remark that all the results discussed in general for asymmetry in the present paper apply directly to coherence, including: the fact that the corresponding robustness measure  is computable via an SDP, that it is directly observable---specifically by considering the analogous notion of coherence witnesses---and that can be endowed with an operational interpretation in terms of channel discrimination \cite{PRL}.
In the rest of this section, we will focus on the derivation of useful technical results and additional analysis which concern specifically the notion of robustness adapted to the special case of coherence. Most of the results obtained in the following are also announced and suitably discussed in \cite{PRL}.

 For completeness, we report the explicit definition of the {\it robustness of coherence} (RoC) $\CR$ of a $d$-dimensional quantum state  $\rho \in {\mathscr{D}}(\mathbb{C}^d)$, adapted from Eq.~(\ref{eq:RoA}) \cite{PRL},
\begin{equation}\label{eq:RoC}
\CR(\rho) = \min_{\tau \in {\mathscr{D}}(\mathbb{C}^d)} \left\{ s\geq 0\ \Big\vert\ \frac{\rho + s\ \tau}{1+s} =: \delta \in {\mathscr{I}}\right\}\,.
\end{equation}
Alternatively, adapting Eq.~\eqref{eq:RoAequivalent}, we can write
\beq
\CR(\rho) = \min_{\sigma \in {\mathscr{I}}} \left\{ s\geq 0\ \Big\vert\ \rho \leq (1+ s)\,\sigma \right\}.
\eeq

\subsection{Bounds on the robustness of coherence}
\label{sec:bounds_coherence}

We first recall that an alternative measure of quantum coherence has been introduced in Ref.~\cite{Baumgratz2014}, namely the $\ell_1$ norm of coherence, defined for a quantum state $\rho$ expanded in the reference basis $\{\ket{j}\}$ as
\begin{equation}\label{eq:CL}
\CL (\rho) = \sum_{ij}|\rho_{ij}|-1 = 2 \sum_{i<j} |\rho_{ij}|\,.
\end{equation}
We have then the following result.

\begin{theorem}
\label{thm:l1bound}
For any state $\rho \in {\mathscr{D}}(\mathbb{C}^d)$ it holds that
\beq
\label{eq:l1bound}
\frac{\CL(\rho)}{d-1} \leq \CR(\rho) \leq \CL(\rho).
\eeq
\end{theorem}

In order to prove Theorem~\ref{thm:l1bound}, we will  make use of the states
\[
\ket{i,j,\theta,\pm}:=\frac{1}{\sqrt{2}}(\ket{i}\pm {\rm e}^{-i\theta}\ket{j})
\]
and of the following lemma.

\begin{lemma}
\label{lem:M}
For any phases $\{\theta_{ij}|i,j=0,\dots,d-1; i<j\}$, the $d\times d$ matrix
\beq
\begin{split}
M(\{\theta_{ij}\})
&=\I+\frac{1}{d-1}\sum_{i<j}({\rm e}^{i\theta_{ij}}\ket{i}\bra{j}+{\rm e}^{-i\theta_{ij}}\ket{j}\bra{i})\\
&=
\begin{pmatrix}
1 						& \frac{{\rm e}^{i\theta_{01}}}{d-1} 	& \cdots 	& \frac{{\rm e}^{i\theta_{0(d-1)}}}{d-1}	\\
\frac{{\rm e}^{-i\theta_{01}}}{d-1}	& 1 						& \cdots  	& \frac{{\rm e}^{i\theta_{1(d-1)}}}{d-1} 	\\
\vdots					& \vdots					& \ddots	& \vdots					\\
\frac{{\rm e}^{-i\theta_{0(d-1)}}}{d-1}	& \frac{{\rm e}^{-i\theta_{1(d-1)}}}{d-1}	& \cdots	& 1
\end{pmatrix}
\end{split}
\eeq
is positive semidefinite, $M(\{\theta_{ij}\})\geq 0$.
\end{lemma}
\begin{proof}
One checks by inspection that the matrix $M(\{\theta_{ij}\})$ can be written as
\[
M(\{\theta_{ij}\})=\frac{2}{d-1}\sum_{i<j} \proj{i,j,\theta_{ij},+},
\]
that is, as the sum of positive semidefinite matrices, hence it is manifestly positive semidefinite.
\end{proof}

\begin{proof} (of Theorem~\ref{thm:l1bound})
To prove the lower bound, consider an optimal incoherent state $\delta^{\star}$ such that
\beq
\label{eq:proofbound}
\rho \leq (1+\CR(\rho))\delta^{\star}.
\eeq
Let $\phi_{ij}$ indicate the phase of the matrix element $\rho_{ij}$, for $i<j$, that is $\rho_{ij}=|\rho_{ij}|{\rm e}^{i\phi_{ij}}$. Consider $M=M(\{-\phi_{ij}\})$, with $M(\{-\phi_{ij}\})$ as in Lemma \ref{lem:M}, with the choice $\{\theta_{ij}\}=\{-\phi_{ij}\}$. The lemma assures that $M\geq 0$. By Schur's theorem,  the Hadamard product---that is, entry-wise product---of two positive semidefinite matrices is positive semidefinite, hence we have
\[
\rho \circ M\leq (1+\CR(\rho))\ \delta^{\star}\circ M.
\]
From the definition of $M$, $\rho'=\rho \circ M$ is a matrix whose off-diagonal entries are the rescaled absolute values of the entries of $\rho$, more precisely $\rho'_{ij}=|\rho_{ij}|/(d-1)$, for $i\neq j$. On the other hand, the diagonal entries of $\rho$ are the same as those of $\rho'$, $\rho'_{ii}=\rho_{ii}$. Furthermore, since $\delta^{\star}$ is diagonal, $\delta^{\star}\circ M = \delta^{\star}$, i.e., still an incoherent state. Therefore,
\[
\rho'\leq (1+\CR(\rho))\delta^{\star}.
\]
We now take  the expectation value on both sides with the maximally coherent state
\begin{equation}\label{eq:maxcoh}
\ket{\psi^+}=\frac{1}{\sqrt{d}}\sum_{j=0}^{d-1}\ket{j}\,,
\end{equation}
obtaining
\[
\bra{\psi^+}\rho'\ket{\psi^+}\leq (1+\CR(\rho))\bra{\psi^+}\delta^{\star}\ket{\psi^+}.
\]
The left-hand side is equal to
\[\frac{1}{d}\left(1+ \frac{2}{d-1}\sum_{i<j} |\rho_{ij}|\right)=\frac{1}{d}\left(1+\frac{1}{d-1}\CL(\rho)\right).\]
On the other hand, since $\delta^{\star} = \sum_{j}\delta^{\star}_{j} \ket{j}\bra{j}$, one has
\[
\bra{\psi^+}\delta^{\star}\ket{\psi^+}=\sum_j \delta^\star_j |\braket{\psi^+}{j}|^2=\sum_j \delta^\star_j \ \frac{1}{d}=\frac{1}{d}.
\]
We then find $\frac{1}{d-1}\CL(\rho) \leq \CR(\rho)$.

To prove the upper bound, we will exhibit a state $\tau$ such that
\[
\frac{\rho+\CL(\rho)\ \tau}{1+\CL(\rho)}
\]
is incoherent. Considering a modification of what was done in~\cite{Robustness,GenRobustness,harrow2003} to calculate the robustness of entanglement of pure states, we define $\tau$ to be
\[
\tau = \frac{2}{\CL(\rho)}\sum_{i<j}|\rho_{ij}|\ \proj{i,j,\phi_{ij},-},
\]
so that
\[
\begin{aligned}
&\quad\rho+\CL(\rho)\ \tau\\
&=\sum_{ij} \rho_{ij} \ket{i}\bra{j}+2\sum_{i<j}|\rho_{ij}|\proj{i,j,\phi_{ij},-}\\
&=\sum_{j}|\rho_{jj}|\ket{j}\bra{j}+\sum_{i<j}|\rho_{ij}|({\rm e}^{i\phi_{ij}}\ket{i}\bra{j}+{\rm e}^{-i\phi_{ij}}\ket{j}\bra{i})\\
&\quad+2 \sum_{i<j}|\rho_{ij}|\proj{i,j,\phi_{ij},-}\\
&=\sum_{j}|\rho_{jj}|\ket{j}\bra{j}\\
&\quad+\sum_{i<j}|\rho_{ij}|(\proj{i,j,\phi_{ij},+}-\proj{i,j,\phi_{ij},-})\\
&\quad+2\sum_{i<j}|\rho_{ij}|\proj{i,j,\phi_{ij},-}\\
&=\sum_{j}|\rho_{jj}|\ket{j}\bra{j}\\
&\quad+\sum_{i<j}|\rho_{ij}|(\proj{i,j,\phi_{ij},+}+\proj{i,j,\phi_{ij},-})\\
&=\sum_{j}|\rho_{jj}|\ket{j}\bra{j}+\sum_{i<j}|\rho_{ij}|(\proj{i}+\proj{j}),
\end{aligned}
\]
which is incoherent. This concludes the proof.
\end{proof}

We argue that the bounds in \eqref{eq:l1bound} are both tight, at least as linear inequalities. The upper bound can be achieved with equality, for instance, on pure $d$-dimensional states, as proven in the next section (see Theorem~\ref{thm:exact}).
 In order to see that the lower bound of \eqref{eq:l1bound} is also tight, it is sufficient to prove that there is a family of states such that the bound is saturated. One such a family is provided by
\begin{equation}\label{eq:rhop}
\rho_p=(1+p)\frac{\I}d-p\proj{\psi^+}, \quad 0\leq p \leq \frac{1}{d-1}.
\end{equation}
For such states, one easily calculates $\CL(\rho_p)=p(d-1)$. On the other hand, $\CR(\rho_p)\leq p$, because  \[\frac{\rho_p+p\proj{\psi^+}}{1+p}=\frac{\I}{d}\] is incoherent.
We thus have that, for $\rho_p$, \[\frac{\CR(\rho_p)}{\CL(\rho_p)}\leq \frac{1}{d-1},\]
which is the opposite of the lower bound of $\eqref{eq:l1bound}$, thus proving that the latter holds with equality.

While the lower bound in \eqref{eq:l1bound} is indeed tight for low values of  $\CL$, one can see that it clearly loosens for larger values of it: indeed, $\CL$ and $\CR$ coincide when assuming their maximal value $d-1$. In order to put this observation on firmer ground, we prove another bound for $\CR$ in terms of $\CL$. We will use that, in the case of coherence, the bounds \eqref{eq:asymmetrybound} become
\beq
\label{eq:coherencelowerbound}
\CR(\rho) \geq \frac{\|\rho-\Delta(\rho)\|_2^2}{\max_{j}\bra{j}\rho\ket{j}}\geq\frac{\|\rho-\Delta(\rho)\|_2^2}{\sqrt{\sum_j \bra{j}\rho\ket{j}^2}}\geq \|\rho-\Delta(\rho)\|_2^2.
\eeq

\begin{theorem}
\label{thm:coherencelowerbound}
For any state $\rho$ it holds that
\begin{equation}\label{eq:lowerbound2}\CR(\rho)\geq f(\CL(\rho),d),
\end{equation}
where
\begin{equation}\label{eq:fCD}
f(C,d)=\frac{dC^2}{(d-1)\Big(-C (d-2)+2 \sqrt{D(C,d)}+d (d-2)+2\Big)},
\end{equation}
with $D(C,d)=(C+1)(d-1)(-C+d-1)$.
\end{theorem}

In order to prove Theorem~\ref{thm:coherencelowerbound} we will need the following Lemma, which may be of independent interest.
\begin{lemma}
\label{lem:boundl1}
Let $p$ a diagonal entry of $\rho \in \mathscr{D}(\bC^d)$. Then
\[
\CL(\rho)\leq \left(\sqrt{p}+\sqrt{1-p}\sqrt{d-1}\right)^2-1.
\]
Inverting the relation, we have that every diagonal entry $p$ of $\rho$ is bounded in terms of $\CL$ by
\[
p\leq \frac{-\CL(\rho) (d-2)+2 \sqrt{D(\CL(\rho),d)}+d^2-2 d+2}{d^2}.
\]
\end{lemma}
\begin{proof}
It is enough to upper bound the $\ell_1$ norm $\|\rho\|_{\ell_1}$ of $\rho$. Without loss of generality, for simplicity we can assume that $p$ is the $\rho_{00}$ entry. Consider first the principal submatrix corresponding to the rows and columns from $1$ to $d-1$. Such a submatrix is $(d-1)\times(d-1)$ and its trace is, by hypothesis, $1-p$. Hence, its $\ell_1$ norm is upper bounded by $(1-p)(d-1)$, which is is achieved by the submatrix whose entries are all equal to $(1-p)/(d-1)$. On the other hand, to evaluate the contribution to $\|\rho\|_{\ell_1}$ of the first row and column, we can focus on the last $d-1$ entries of row $0$, since such a contribution is equal to $p+2\sum_{j=1}^{d-1} |\rho_{0j}|$ due to the hermiticity of $\rho$. The positivity of $\rho$ forces $|\rho_{0j}|\leq \sqrt{\rho_{00}\rho_{jj}}=\sqrt{p}\sqrt{\rho_{jj}}$. Therefore
\begin{eqnarray}
\sum_{j=1}^{d-1} |\rho_{0j}| \leq \sum_{j=1}^{d-1}\sqrt{p}\sqrt{\rho_{jj}} &\leq & \sqrt{d-1}\sqrt{p} \sqrt{\sum_{j=1}^{d-1}\rho_{jj}} \nonumber \\ &=& \sqrt{d-1}\sqrt{p(1-p)}, \nonumber
\end{eqnarray}
by the Cauchy-Schwartz inequality (equivalently, by the concavity of the square root) and the fact that $\sum_{j=1}^{d-1} p_{jj}=1-\rho_{00}=1-p$. The bound is saturated by the choice $\rho_{0j}=\sqrt{\frac{p(1-p)}{d-1}}$, consistent---in terms of positivity of the overall matrix---with the choice $\rho_{jj}=\frac{1-p}{d-1}$ for $1\leq j \leq d-1$. Overall we found
\begin{eqnarray}
\|\rho\|_{\ell_1} &\leq& p + 2\sqrt{d-1}\sqrt{p(1-p)}+(1-p)(d-1) \nonumber \\
&\leq& \big(\sqrt{p}+\sqrt{1-p}\sqrt{d-1}\big)^2,\nonumber
\end{eqnarray}
which, as evident through our construction, is saturated by the density matrix corresponding to the pure state $\ket{\psi}=\sqrt{p}\ket{0}+\sqrt{\frac{1-p}{d-1}}\sum_{j=1}^{d-1}\ket{j}$.
\end{proof}

We are now ready to prove Theorem~\ref{thm:coherencelowerbound}.
\begin{proof} (of Theorem~\ref{thm:coherencelowerbound}) We have
\[
\|\rho-\Delta(\rho)\|_2^2=2\sum_{i<j}|\rho_{ij}|^2\geq \frac{4}{d(d-1)} \left(\sum_{i<j}|\rho_{ij}|\right)^2=\frac{\mathscr{C}^2_{\ell_1}(\rho)}{d(d-1)}.
\]
From \eqref{eq:coherencelowerbound} we have then
\[
\CR(\rho) \geq \frac{\mathscr{C}^2_{\ell_1}(\rho)}{d(d-1)\max_i \bra{i}\rho\ket{i}}.
\]
We can now invoke Lemma \ref{lem:boundl1}, which, applied in the case $p=\max_i \bra{i}\rho\ket{i}$,
lets us conclude
\[
\CR(\rho) \geq f(\CL(\rho),d),
\]
with $f(C,d)$ defined in Eq.~(\ref{eq:fCD}).
\end{proof}

We remark that $f(C,d)$ as in Theorem~\ref{thm:coherencelowerbound} is continuous and satisfies $f(d-1,d)=d-1$, which proves that $\CR(\rho)$ converges to $\CL(\rho)$ for all $d$-dimensional states $\rho$ with $\CL(\rho)$ close to its maximum value $d-1$.

\begin{figure*}[t!]
  \centering
\includegraphics[width=8cm]{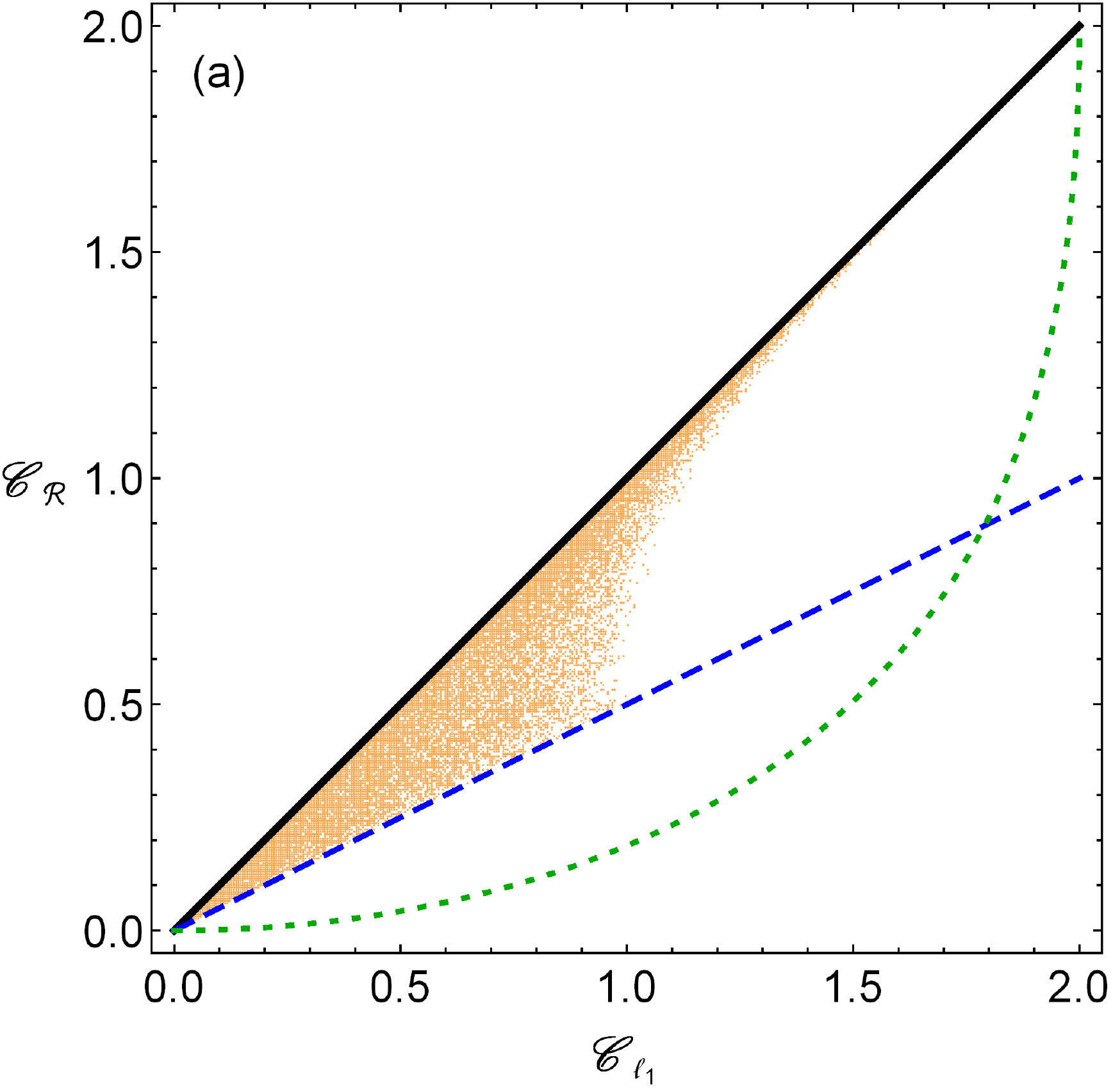}\hspace{1cm}
\includegraphics[width=8cm]{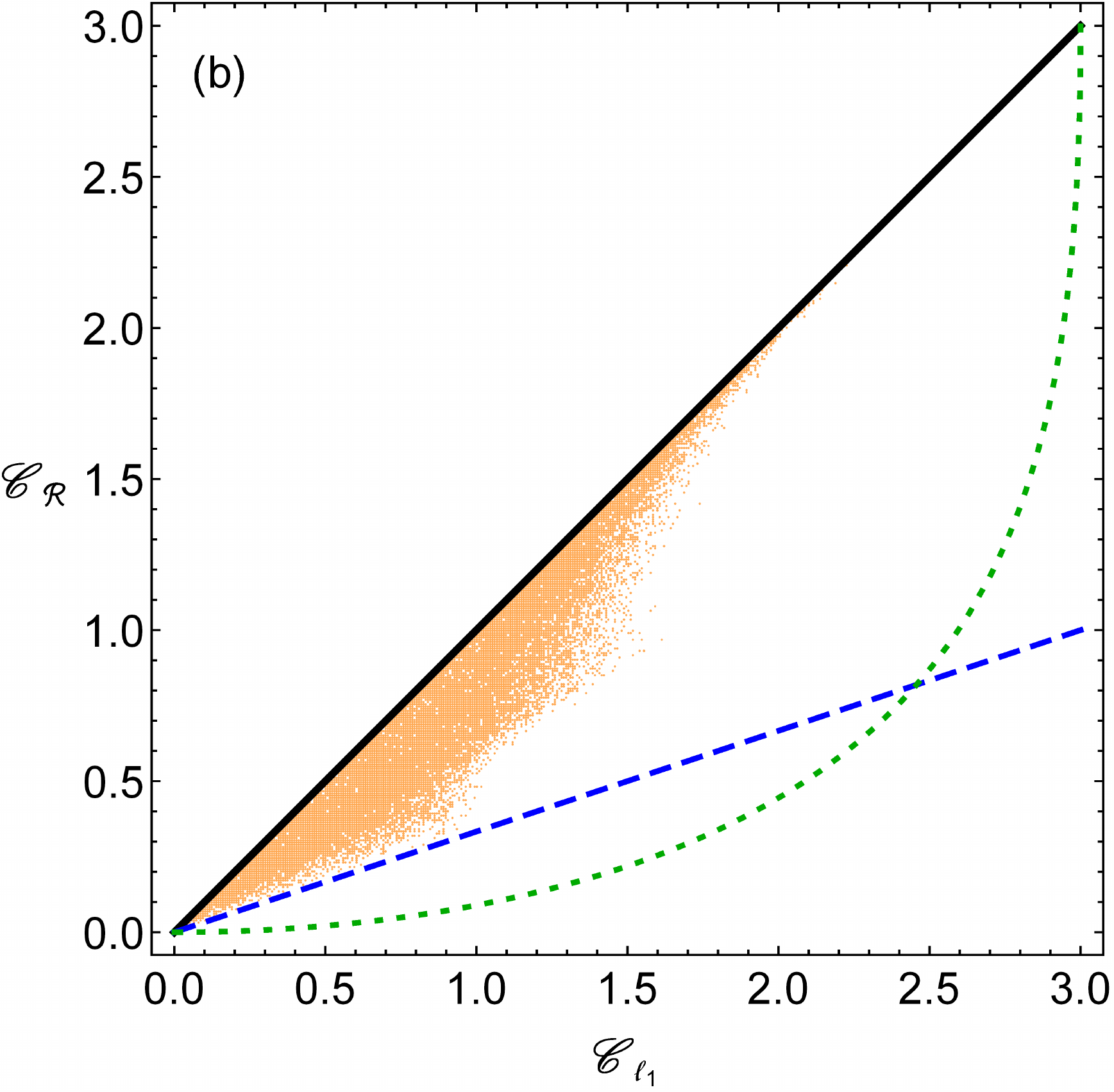} \\
 \caption{
(Color online) Comparison between the $\ell_1$ norm of coherence $\CL$ (horizontal axis) and the robustness of coherence $\CR$ (vertical axis) for $3 \times 10^4$ randomly generated $d$-dimensional states, with  (a) $d=3$, and (b) $d=4$. In all panels we plot additionally some bounds to $\CR$ as a function of $\CL$: the solid thick black line $\CR=\CL$ denotes the upper bound of Eq.~(\ref{eq:l1bound}), which is saturated by pure states; the dashed blue line $\CR=\CL/(d-1)$ denotes the lower bound of Eq.~(\ref{eq:l1bound}), which is saturated by the states in (\ref{eq:rhop}) up to $\CL = 1$; the dotted green line denotes the alternative (non-tight) lower bound of Eq.~(\ref{eq:lowerbound2}). All the quantities plotted are dimensionless.}
 \label{fig:lowers}
\end{figure*}

In Fig.~\ref{fig:lowers} we compare the measures $\CR$ versus $\CL$ for randomly generated states in dimension $d=3$ and $d=4$. Notice that for all states of a qubit ($d=2$) the two measures coincide instead,  as remarked in the next section.

\subsection{Exact robustness of coherence for pure states and generalized X states}
\label{sec:exact_coherence}

Here we show that for a relevant class of mixed states of $d$-dimensional systems one can evaluate the RoC exactly \cite{PRL}.

\begin{theorem}
\label{thm:exact}

Let $\rho \in {\mathscr{D}}(\mathbb{C}^d)$ be a state such that there exists a unitary $U = \sum_j {\rm e}^{i \phi_j} \ket{j}\!\bra{j}$, diagonal in the reference basis $\{\ket{j}\}$, which maps $\rho$ into $\rho' = U\rho U^{\dagger}$ with entries  $\rho'_{ij} = |\rho_{ij}|$. Then
\begin{equation}\label{eq:Rex}
\CR(\rho)=\CL(\rho).
\end{equation}
In particular, for the RoC of a pure state $\ket{\psi} \in \bC^d$,
\begin{equation}\label{eq:purestate}
\ket{\psi}=\sum_j \psi_j \ket{j},
\end{equation}
this gives  the result
\begin{equation}\label{eq:Rpure}
\CR(\proj{\psi})=\CL(\proj{\psi})=\mbox{$\left(\sum_j |\psi_j|\right)^2-1$}.
\end{equation}
\end{theorem}

\begin{proof}
We can invoke the bound \eqref{eq:l1bound} $\CR(\rho)\leq \CL(\rho)$. To prove that this is an equality under the conditions of the theorem, we will  tighten the general lower bound of Theorem~\ref{thm:l1bound}. Indeed, one can adapt the proof of the lower bound of Theorem~\ref{thm:l1bound}, by considering $\rho'=U\rho U^{\dagger}$, instead of $\rho'=\rho\circ M$,  so that $\bra{\psi^+} \rho' \ket{\psi^+}=\frac{1}{d}\sum_{ij}|\rho_{ij}|$. The last steps in the proof of the lower bound are then the same.

That a pure state $\ket{\psi}$ admits such  unitary it is clear: take
\[
U=\sum_j {\rm e}^{-i\phi_j}\proj{j},
\]
 where $\phi_j$ is the phase of the coefficients $\psi_j$, i.e., $\psi_j=|\psi_j|{\rm e}^{i\phi_j}$.
\end{proof}

It is clear that the bounds \eqref{eq:l1bound} give the exact value of $\CR$ for an arbitrary state of one qubit ($d=2$), for which $\CR$ is then equal to $2|\rho_{01}|$ \cite{PRL}. Also, a qubit state is one such that there exists a unitary $U$ as in Theorem~\ref{thm:exact}.

Of course, all states such that their entries in the reference basis are positive to begin with (that is, such that we can take $U=\I$ in Theorem~\ref{thm:exact}) satisfy $\CR=\CL$.
A simple class of states for which Theorem~\ref{thm:exact} holds less trivially are generalized X-states, of the form
\[
\rho
=
\begin{cases}
\sum_{j=0}^{\frac{d}{2}}\rho_{j}&\textrm{if $d$ is even;}\\
\sum_{j=0}^{\lfloor \frac{d}{2}\rfloor}\rho_{j}+\rho_c&\textrm{if $d$ is odd,}
\end{cases}
\]
with
\[
\begin{aligned}
\rho_j &=  \rho_{jj}\proj{j}+\rho_{j,d-1-j}\ket{j}\bra{d-1-j}\\
&\quad+\rho_{d-1-j,j}\ket{d-1-j}\bra{j}+\rho_{d-1-j,d-1-j}\proj{d-1-j}
\end{aligned}
\]
and
\[
\rho_c= \rho_{\lfloor d/2\rfloor+1}\proj{\lfloor d/2\rfloor+1}.
\]
Such a class comprises, in the bipartite case, all two-qubit X-states, a superclass of Bell diagonal states (see~\cite{rau2009algebraic} and references therein).


%
%

\subsection{Only maximally coherent states have maximal robustness of coherence}

Here we prove that the RoC is a measure of coherence whose maximal value can only be reached on pure maximally coherent states of the form (\ref{eq:maxcoh}). This is a desired property for a valid measure of coherence \cite{MaxCoh}.

\begin{theorem}
A state $\rho \in {\mathscr{D}}(\mathbb{C}^d)$ satisfies $\CR(\rho)=d-1$ if and only if $\CR(\rho)$ is a maximally coherent pure state.
\end{theorem}

\begin{proof} With a more careful analysis of the proof of Property~\ref{A1}, we observe that for every state $\rho$ it holds that
\[
\rho\leq\|\rho\|_\infty \I = d\|\rho\|_\infty \frac{\I}{d},
\]
where $\|\rho\|_\infty \leq 1$ is the largest eigenvalue of $\rho$. Since the maximally mixed state $\frac{\I}{d}$ is incoherent, this implies
\[
\CR(\rho) \leq d\|\rho\|_\infty - 1.
\]
If $\CR(\rho) = d-1$, the just found inequality implies $\|\rho\|_\infty =1$, that is, that $\rho$ is a pure state.
We can now invoke the result that for pure states the RoC is equal to the $\ell_1$ norm of coherence (Theorem~{thm:exact}), and the fact that the only pure states with maximal $\ell_1$ norm of coherence are are the maximally coherent states \cite{Baumgratz2014}, to conclude the proof.
\end{proof}

\section{Conclusions}

The importance of symmetry in physics can hardly be overestimated. A quantum state may or may not respect a given symmetry; in the latter case one says that the state is asymmetric. With the advent of quantum information processing and its operational approach to quantum features, the asymmetry of states has been investigated more rigorously and elevated to the status of a resource. In this paper, which also acts as a companion for the Letter~\cite{PRL}, we have introduced explicitly the robustness of asymmetry, a quantifier of asymmetry that has been shown to possess several desirable properties, including: a defining operational interpretation as resilience against noise of the asymmetry present in a given state; a further operational characterization in the context of of channel discrimination, in particular in terms of the advantage that an asymmetric state can provide in the discrimination of the channels that realize the representation of the symmetry group under consideration; an efficient numerical evaluation via semidefinite programming, once the state is known; the possibility to be measured or at least estimated directly experimentally, hence yielding a convenient benchmark for non-classicality in disparate physical scenarios. Furthermore, the robustness of asymmetry has been shown to be an asymmetry monotone in a strong sense, hence it can be employed as quantifier of asymmetry in a variety of resource theoretic frameworks in which, while the notion of free states identified as symmetric states remains the same, the notion of free operations adopted to manipulate asymmetry may differ.

Quantum coherence can be considered, from an operational point of view, as a special case of asymmetry. Consequently, all the tools we introduced and developed for the study of asymmetry immediately specialize to coherence, including in particular the notions of robustness of coherence and of coherence witnesses. While the relevance of the latter concepts is emphasized in the companion Letter~\cite{PRL}, in this paper we have provided full details and proofs for the claims made there.

As a service to the community, we provide numerical code for the evaluation and estimation of both the robustness of asymmetry and the robustness of coherence as Supplemental Material \cite{epapsA}. We expect that the concepts and tools---be them analytical or numerical---that we developed for the study of both asymmetry and, in particular, coherence, will be helpful for further theoretical developments and improved understanding of such fundamental concepts, and also for their experimental verification and benchmarking of quantum behaviour in physical and biological domains.

The present paper, together with \cite{PRL}, stands as further evidence that a modern quantum information approach to basic concepts, like symmetry/asymmetry and coherence, can shed further light on them and contribute to put them on solid qualitative and quantitative grounds, highlighting their role and usefulness in fundamental and technological applications.

\section*{Ackowledgements}

This work has received funding from the European Research Council (ERC StG GQCOP, Grant No.~637352),
from the European Union's Horizon 2020 research and innovation programme under the Marie Sklodowska-Curie (Grant Agreement No.~661338),
and the Erasmus+ Programme. We thank D.~Girolami, M.~Hall, I.~Marvian, M.~B.~Plenio, A.~Streltsov, and J.~Vaccaro for useful discussions.

\bibliographystyle{apsrevfixed}
\bibliography{corrub}

\end{document}